\documentclass{sig-alternate}
\usepackage{times, amsmath, alltt, amssymb, xspace, setspace,epsfig}
\usepackage{psfrag}
\usepackage{algorithm, algorithmic}

\newtheorem{theorem}{Theorem}

\newtheorem{prop}[theorem]{Proposition}
\newtheorem{cor}[theorem]{Corollary}
\newtheorem{definition}{Definition}

\newtheorem{alg}{Algorithm}

\newcommand\comment[1]{}

\newcommand\todo[1]{\{\textbf{Todo:} \textit{#1}\}}

\newcommand{\DP}[1]{\ensuremath{#1\mbox{-}\mathsf{DP}}\xspace}

\newcommand{\DPS}[1]{\ensuremath{#1\mbox{-}\mathsf{DPS}}\xspace}
\newcommand{\bedDPS}{\DPS{(\beta,\epsilon,\delta)}}

\newcommand{\eDP}{\DP{\epsilon}}
\newcommand{\edDP}{\DP{(\epsilon,\delta)}}

\newcommand{\NRI}[1]{\ensuremath{#1\mbox{-}\mathsf{NRI}}\xspace}

\newcommand{\edNRI}[1]{\NRI{(\epsilon,\delta)}}

\renewcommand{\AA}{\ensuremath{\mathcal{A}}\xspace}

\newcommand{\mypara}[1]{\vspace*{0.05in}\noindent\textbf{#1} \xspace}

\xspace

\newcommand{\odds}{\ensuremath{\mathsf{odds}}\xspace}

\begin{document}

\hyphenation{identfi-cation supp-ression}

\title{On Sampling, Anonymization, and Differential Privacy: Or, {\ttlit{k}}-Anonymization Meets Differential Privacy}

\numberofauthors{3}

\author{
%
%
\alignauthor
Ninghui Li\\
       \affaddr{Purdue University}\\
       \affaddr{305 N. University Street,}
       \affaddr{West Lafayette, IN 47907, USA}\\
       \email{ninghui@cs.purdue.edu}
\alignauthor
Wahbeh Qardaji\\
       \affaddr{Purdue University}\\
       \affaddr{305 N. University Street,}
       \affaddr{West Lafayette, IN 47907, USA}\\
       \email{wqardaji@cs.purdue.edu}
\alignauthor
Dong Su\\
       \affaddr{Purdue University}\\
       \affaddr{305 N. University Street,}
       \affaddr{West Lafayette, IN 47907, USA}\\
       \email{su17@cs.purdue.edu}
}

\maketitle

\pagenumbering{arabic}

\sloppypar

\begin{abstract}
This paper aims at answering the following two questions in privacy-preserving data analysis and publishing: What formal privacy guarantee (if any) does $k$-anonymization provide?  How to benefit from the adversary's uncertainty about the data?
We have found that random sampling provides a connection that helps answer these two questions, as sampling can create uncertainty.  The main result of the paper is that $k$-anonymization, when done ``safely'', and when preceded with a random sampling step, satisfies $(\epsilon,\delta)$-differential privacy with reasonable parameters.  This result illustrates that ``hiding in a crowd of $k$'' indeed offers some privacy guarantees.
This result also suggests an alternative approach to output perturbation for satisfying differential privacy: namely, adding a random sampling step in the beginning and pruning results that are too sensitive to change of a single tuple.
Regarding the second question, we provide both positive and negative results.  On the positive side, we show that adding a random-sampling pre-processing step to a differentially-private algorithm can greatly amplify the level of privacy protection.  Hence, when given a dataset resulted from sampling, one can utilize a much large privacy budget.  On the negative side, any privacy notion that takes advantage of the adversary's uncertainty likely does not compose.  We discuss what these results imply in practice.
\end{abstract}

\section{Introduction}


In this paper we deal with the problem of using data in a privacy-preserving way.  We consider the scenario where a trusted curator obtains a dataset by gathering private information from a large number of respondents, and then make usage of the dataset while protecting the privacy of respondents.  The curator may learn and release to the public statistical facts about the underlying population.
Alternatively, the curator may publish a sanitized (or, ``anonymized'') version of the dataset so that other parties can use the data to perform any analysis they are interested in.

This paper aims at answering the following two questions in privacy-preserving data analysis and publishing.  The first is: What formal privacy guarantee (if any) does $k$-anonymization methods provide?  $k$-Anonymization methods have been studied extensively in the database community, but have been known to lack strong privacy guarantees.  The second question is: How to benefit from the adversary's uncertainty about the data?
More specifically, can we come up a meaningful relaxation of differential privacy~\cite{Dwo06,Dwo08} by exploiting the adversary's uncertainty about the dataset?  We now discuss these two motivations in more details.


The $k$-anonymity notion was introduced by Sweeny and Samarati~\cite{Swe02a, Swe02b, Samarati01, SS98a} for privacy-preserving microdata publishing.  This notion has been very influential.  Many $k$-anonymization methods have been developed over the last decades; it has also been extensively applied to other problems such as location privacy~\cite{GBL08}.  The $k$-anonymity notion requires that when only certain attributes, known as quasi-identifiers (QIDs), are considered, each tuple in a $k$-anonymized dataset should appear at least $k$ times.  In this paper, we consider a version of $k$-anonymity which treats all attributes as QIDs.
We show that even satisfying this strong version of $k$-anonymity does not protect against re-identification attacks. In addition, we identify the privacy vulnerabilities of existing $k$-anonymization algorithms.  We then define classes of $k$-anonymization algorithms that are ``strongly-safe'' and ``$\epsilon$-safe'', which avoid the privacy vulnerabilities of existing $k$-anonymization algorithms.  The question we aim to answer is whether these safe $k$-anonymization methods would provide strong enough privacy guarantee in practice.

The notion of differential privacy was introduced by Dwork et al.~\cite{Dwo06,DMN+06}.
An algorithm \AA satisfies $\epsilon$-Differential Privacy (\eDP) if and only if for any two neighboring datasets $D$ and $D'$, the distributions of $\AA(D)$ and $\AA(D')$ differ at most by a multiplicative factor of $e^\epsilon$.
A relaxed version of \eDP, which we use \edDP to denote, allows an error probability bounded by $\delta$.
Satisfying differential privacy ensures that even if the adversary has full knowledge of the values of a tuple $t$, as well as full knowledge of what other tuples are in the dataset, and is only uncertain about whether $t$ is in the input dataset, the adversary cannot tell whether $t$ is in the dataset or not beyond a certain confidence level.  As in most data publishing scenarios, the adversary is unlikely to have precise information about all other tuples in a dataset.  It is desirable to exploit this uncertainty to define a relaxed version of differential privacy, which can be easier to satisfy.


We have found that sampling provides the link between our two goals.  The main result in this paper is that sampling plus ``safe'' $k$-anonymization satisfies \edDP.  This result leads us to study the relationship between sampling and differential privacy.  We say that an algorithm satisfies differential privacy under sampling if the algorithm preceded with a random sampling step satisfies differential privacy.

Results about differential privacy under sampling both are of theoretical interest and have practical relevance. Sampling is a natural way to model the adversary's uncertainty about the data; thus this helps understand how to take advantage of this uncertainty in private data analysis.
On the practical side, many data publishing scenarios already involve a random sampling step.  Sometimes this sampling step is explicit, when one has a large dataset and wishes to release only a much smaller for research, such as the US census bureau's 1-percent Public Use Microdata Sample.
Sometimes, this sampling step is implicit; because the respondents are randomly selected, one can view the dataset as resulted from sampling.

The contributions of this paper are as follows:

\begin{itemize}
 \item
We prove that safe $k$-anonymization algorithm, when preceded by a random sampling step, provides $(\epsilon, \delta)$-differential privacy with reasonable parameters.

In the literature, $k$-anonymization and differential privacy have been viewed as very different privacy guarantees: $k$-anonymization is syntactic, and differential privacy is algorithmic and provides semantic privacy guarantees.  Our result is, to our knowledge, the first to link $k$-anonymization with differential privacy.  It illustrates that ``hiding in a crowd of $k$'' indeed offers privacy guarantees.

This result also provides a new way of satisfying differential privacy.  Existing techniques for satisfying differential privacy rely on output perturbation, that is, adding noise to the query outputs.  Our result suggests an alternative approach.  Rather than adding noise to the output, one can add a random sampling step in the beginning and prune results that are too sensitive to changes of individual tuples (i.e., tuples that violate $k$-anonymity).

 \item
We show both positive and negative results on utilizing the adversary's uncertainty about the data.  On the positive side, we show that random sampling has a privacy amplification effect for \edDP.  For an algorithm that satisfies \edDP, adding a sampling step with probability $\beta$ reduces both $e^{\epsilon}-1$ and $\delta$ by a factor of $\beta$.  For example, applying an algorithm that achieves ($\ln 2 \approx 0.69$)-differential privacy on dataset sampled with $0.1$ probability can achieve overall ($\ln 1.1\approx 0.095$)-differential privacy.

On the negative side, we show that any privacy notion that exploits the adversary's uncertainty about the data is unlikely to compose, in the sense that publishing the output from two algorithms together may be non-private.

Our results suggest the following approaches to take advantage of the fact that the input dataset is resulted from explicit or implicit sampling.  If one applies algorithms that satisfy \edDP, then one can allow a larger privacy budget because of sampling.  If one applies an algorithm that does not satisfy \edDP, but satisfies \edDP under sampling, then it is safe to apply the algorithm once.  However, if one has a large dataset, one can repeated sample and then apply the algorithm on each newly sampled dataset.

\end{itemize}


The rest of the paper is organized as follows.  We study the relationship between differential privacy and sampling in Section~\ref{sec:dps}.  We study $k$-anonymization and prove our main result in Section~\ref{sec:result}.
We discuss related work in Section~\ref{sec:related} and conclude in Section~\ref{sec:conclusions}.  An appendix includes proofs not found in the main body.

\section{Differential Privacy Under Sampling}\label{sec:dps}


\subsection{Differential Privacy}~\label{sec:differential}


Differential privacy formalizes the following protection objective: if a disclosure occurs when an individual participates in the database, then the same disclosure also occurs with similar probability (within a small multiplicative factor) even when the individual does not participate.  More formally, differential privacy requires that, given two input datasets that differ only in one tuple, the output distributions of the algorithm on these two datasets should be close.


\begin{definition}\textbf{[$\epsilon$-Differential Privacy~\cite{Dwo06,DMN+06} (\eDP)]:} \label{def:diff}
A randomized algorithm $\AA$ gives $\epsilon$-differential privacy if for any pair of neighboring datasets $D$ and $D'$, and any $O\subseteq  \mathsf{Range}(\AA)$,
\begin{equation}\label{eqn:dp}
\Pr[\AA(D)\in O] \leq e^{\epsilon}\Pr[\AA(D') \in O]
\end{equation}
\end{definition}

Intuitively, \eDP offers strong privacy protection.  If $\AA$ satisfies \eDP, one can claim that publishing $\AA(D)$ does not violate the privacy of any tuple $t$ in $D$, because even if one leaves $t$ out of the dataset, in which case the privacy of $t$ can be considered to be protected, one may still publish the same outputs with a similar probability.

\comment{
Differential privacy provides worst-case privacy guarantee in at least two senses.  First, inequality~(\ref{eqn:dp}) must hold for \emph{every} dataset $D$ and \emph{every} tuple $t$ in $D$, meaning that the privacy for every tuple in every dataset is protected.  Second, inequality~(\ref{eqn:dp}) requires that the $e^{\epsilon}$ bound holds for \emph{every} possible outcome $S$, even if $S$ occurs only with very low probability.
These are desirable features, since as we have seen in past privacy incidents, compromising the privacy of a single individual may already be viewed unacceptable, and should be avoided~\cite{BZ06,NS08}.
}

In practice, \eDP can be too strong to satisfy in some scenarios.  A commonly used relaxation is to allow a small error probability $\delta$.
\begin{definition}\textbf{[$(\epsilon, \delta)$-Differential Privacy~\cite{DKM+06} (\edDP)]:} \label{def:edDP}
A randomized algorithm $\AA$ satisfies $(\epsilon,\delta)$-differential privacy, if for any pair of neighboring datasets $D$ and $D'$ and for any $O\subseteq \mathsf{Range}(\AA)$:
$$
Pr[\AA(D)\in O] \leq e^\epsilon Pr[\AA(D')\in O] + \delta
$$
\end{definition}


Existing methods to satisfy differential privacy includes adding Laplace noise proportional to the query's global sensitivity~\cite{Dwo06,DMN+06}, adding noise related to the smooth bound of the query's local sensitivity~\cite{NRS07}, and the exponential mechanism to select a result among all possible results~\cite{MT07}.

\subsection{Uncertain Background Knowledge} \label{sec:dps:uncertain}

One of our goals is to develop a further relaxation of differential privacy that can be more easily satisfied.  The intuition that we wanted to exploit is the adversary's uncertainty about the underlying dataset.  The \edDP notion ensures that when an adversary is uncertain about whether one tuple $t$ is present in the input dataset, even when the adversary knows the \emph{precise information} all other tuples in the input dataset, the adversary cannot tell based on the output whether $t$ is in the input or not.  We believe that it is reasonable to relax the assumption to that the adversary knows all attributes of a tuple $t$ (but not whether $t$ is in the dataset), and in addition statistical information about the rest of the dataset $D$.  The privacy notion should prevent such an adversary from substantially distinguishing between $D$ and $D\cup \{t\}$ based on the output.



The desire to exploit adversary's uncertainty is shared by other researchers.  For example, Adam Smith's blog post summarizing the Workshop on Statistical and Learning-Theoretic Challenges in Data Privacy includes a section on relaxed definitions of privacy with meaningful semantics: ``it would be nice to see meaningful definitions of privacy in statistical databases that exploit the adversary's uncertainty about the data.  The normal approach to this is to specify a set of allowable prior distributions on the data (from the adversary's point of view).  However, one has to be careful.  The versions I have seen are quite brittle.''\footnote{\textit{http://adamdsmith.wordpress.com/2010/03/04/ipam-workshop-wrap-up/}}




Some degree of brittleness may be unavoidable.  It appears that any privacy notion that takes advantage of the adversary's uncertainty about the data is not robust under composition, which requires that given two algorithms that both satisfy the privacy notion, their composition, i.e., applying both algorithms to the same input dataset and then publish both outputs, also satisfies the privacy notion.

Consider the following two algorithms.  Let $r(D)$ be the predicate that $D$ contains an odd number of tuples, and $s(D)$ be a sensitive predicate, e.g., whether a tuple $t$ is in $D$.  Algorithm $\AA_1(D)$ outputs $r(D)$, and $\AA_2(D)$ outputs $r(D)$ XOR $s(D)$.  Both $\AA_1$ and $\AA_2$ should satisfy a privacy notion that assumes that the adversary is uncertain about the data, because there is no reason that the adversary should know the exact number of the tuples.  However, the composition of $\AA_1$ and $\AA_2$ leaks $r(D)$.  More generally, for any privacy  definition that exploits the adversary's uncertainty about data, there exists at least one predicate that the adversary is uncertain about.  Then one algorithm can output that predicate, and a second algorithm can output that predicate XOR's with a predicate that results in privacy leakage; and they does not compose.

The above observation suggests that no such definition should be used in the interactive setting of answering multiple queries.  If, however, one intends to publish a dataset in the non-interactive setting only once, then the inability to compose may be an acceptable limitation.

\subsection{Differential Privacy under Sampling}

One natural approach to capturing the adversary's uncertainty about the input data is to add a sampling step.  We introduce the following definition, called $(\beta,\epsilon,\delta)$-Differential Privacy under Sampling (\bedDPS for short).

\begin{definition}[Differential privacy under sampling]\label{def:DPS}
An algorithm $\AA$ gives \bedDPS if and only if $\beta>\delta$ and the algorithm $\AA^\beta$ gives \edDP, where $\AA^\beta$ denotes the algorithm to first sample with probability $\beta$ (include each tuple in the input dataset with probability $\beta$), and then apply \AA to the sampled dataset.  
\end{definition}



The above definition requires $\beta>\delta$ because any algorithm trivially satisfies \DPS{(\beta,0,\delta)} when $\beta \leq \delta$.  This is because when two datasets differ only by one tuple, sampling from them with the probability $\beta$ will result in exactly the same output with probability $1-\beta$.  However, when $\beta \gg \delta$, the notion of \bedDPS is both nontrivial to satisfy and a nontrivial relaxation of \edDP, as shown by our results in Section~\ref{sec:result}.  There we show that existing $k$-anonymization algorithms do not satisfy \bedDPS, and have privacy vulnerabilities, and that safe (and possibly deterministic) $k$-anonymization satisfies \bedDPS, while violating \edDP for any $\delta<1$.

%
%
%

\subsection{The Amplification Effect of Sampling}\label{sec:udp:amp}

An interesting feature of the \bedDPS notion is that there is a connection between the privacy parameters $\epsilon,\delta$ and the sampling rate $\beta$.  The following theorem shows that by employing a smaller sampling rate, one can achieve a stronger privacy protection (i.e., smaller values for $\epsilon$ and $\delta$).

\begin{theorem}\label{thm:amplification}
Any algorithm that satisfies $\DPS{(\beta_1,\epsilon_1,\delta_1)}$ also satisfies $\DPS{(\beta_2,\epsilon_2,\delta_2)}$ for any $\beta_2 < \beta_1$, where \\ $\epsilon_2=\ln\left(1+\left(\frac{\beta_2}{\beta_1}(e^{\epsilon_1}-1)\right)\right)$, and $\delta_2=\frac{\beta_2}{\beta_1}\delta_1$.
\end{theorem}
See Appendix~\ref{app:amplification} for the proof.

An equivalent way to write $\epsilon_2=\ln\left(1+\left(\frac{\beta_2}{\beta_1}(e^{\epsilon_1}-1)\right)\right)$ is
$$\frac{e^{\epsilon_2}-1}{e^{\epsilon_1}-1} = \frac{\beta_2}{\beta_1}.$$
In other words, by decreasing the sampling probability, one obtains proportional decreases in $e^{\epsilon}-1$ and $\delta$, improving the privacy protection.  Hence, when one possesses a randomly sampled dataset, then one can use much relaxed privacy budget $\epsilon$ and error toleration $\delta$.  To see the effects of this, in Table~\ref{table:sampling} we show the privacy parameters for an algorithm that satisfies $\DP{(\ln 11, 10^{-5})}$, and an algorithm that satisfies $\DP{(1,0)}$ under sampling rate $0.1$ and $0.01$.

\begin{table}
\begin{center}
\begin{tabular}{lr}
$
\begin{array}{|c|c|c|c|}
\hline
\beta & e^{\epsilon} & \epsilon & \delta
\\ \hline
1 & 11 & \ln 11\approx 2.40   & 10^{-5}
\\
0.1 &  2 & \ln 2\approx 0.69 & 10^{-6}
\\
0.01 & 1.1 & \ln 1.1 \approx 0.095 & 10^{-7}
\\ \hline
\end{array}
$
&
$
\begin{array}{|c|c|}
\hline
\beta & \epsilon
\\ \hline
1 & 1
\\
0.1 & 0.159
\\
0.01 & 0.017
\\ \hline
\end{array}
$
\end{tabular}
\end{center}
\caption{Effect of privacy parameters under sampling.}\label{table:sampling}
\end{table}

\comment{
Chaudhuri and Mishra \cite{CM06} used the random sampling in a different way to achieve $(1,\epsilon, \delta)$-privacy, another relaxation of differential privacy. First suppress those rare tuples in the dataset to form a safe dataset, then do random sampling with a rate $\beta$ on a such dataset. They proved that if the sampling rate is at most $O(\epsilon)$, the whole process can achieve $(1,O(\epsilon),\delta)$-privacy. But this approach assumes that there are few rare values in the original dataset. This assumption makes it difficult to apply this approach in publishing a dataset with distinct tuples, such as the search log dataset where every user's searche histories are almost distinct from others'. Another limitation is the result the sampling rate $\beta$ varies linearly as the failure probability $\delta$. But as Korolova et al. \cite{KKMN09} suggested, the typical value of failure probability should be $O(1/\mbox{number of tuples})$ which will make the above sampling probaility $\beta$ useless.
}

Smith's blog~\footnote{\emph{http://adamdsmith.wordpress.com/2009/09/02/sample-secrecy/}} includes an ``amplification'' lemma for differential privacy, which was used implicitly in the design of a PAC learner for the parity class in~\cite{KLN+08}.  The lemma states that an algorithm that satisfies $\DP{(\epsilon=1)}$, when preceded by random sampling with rate $\beta$, satisfies $\DP{(2\beta)}$.  Theorem~\ref{thm:amplification} exploits similar observations, but is more general in that it applies to \edDP, rather than \eDP, and that it also applies to arbitrary values of $\epsilon$.  Our result is also slightly tighter; for example, for the special case of $\epsilon=1$ and $\beta=0.1$, we give a result of $0.159$ as opposed to $2\beta=0.2$.

\subsection{Properties of \bedDPS}

While the \bedDPS notion does not compose.  It does have several other desirable properties.  In~\cite{KL10}, Kifer and Lin identified two privacy axioms when they defined the generic differential privacy.  The \emph{Transformation Invariance} axiom states that given an algorithm $\AA$ that satisfies a privacy notion, adding any post-processing step operating on $\AA$'s output should still satisfy the privacy notion.  The \emph{Privacy Axiom of Choice} axiom states that given two algorithms $\AA_1$ and $\AA_2$ that both satisfy a privacy notion, then a new algorithm that chooses $\AA_1$ with probability $p$ and $\AA_2$ with probability $1-p$ should also satisfy the notion.
We now show that \bedDPS satisfies both axioms.

\begin{theorem}\label{thm:postprocessing}
Given $\AA_1$ that satisfies $\DPS{(\beta,\epsilon,\delta)}$ and any algorithm $\AA_2$, $\AA(D)=\AA_2(\AA_1(D))$ satisfies $\DPS{(\beta,\epsilon,\delta)}$.
\end{theorem}
\begin{proof}
Assume, for the sake of contradiction, that $\AA(D)$ does not satisfy \bedDPS, then there exist neighboring $D$ and $D'$ and
$O\subseteq \mathsf{Range}(\AA_2)$ such that
$$\Pr[\AA_2(\AA_1^\beta(D)) \in O] > e^{\epsilon}  \Pr[\AA_2(\AA_1^\beta(D')) \in O] + \delta$$

Consider all $S$'s in $\mathsf{Range}(\AA_1)$, let $q(S)=\Pr[\AA_2(S)\in O]$, and let $p(S)=\Pr[\AA_1^\beta(D)=S]$ and $p'(S)=\Pr[\AA_1^\beta(D')=S]$. Then we have
$$\sum_{S\in \mathsf{Range}(\AA_1)} p(S)q(S) > e^\epsilon \sum_{S\in \mathsf{Range}(\AA_1)} p'(S)q(S) + \delta.$$

We partition $\mathsf{Range}(\AA_1)$ into
$\mathcal{S}_1=\{ S \mid p(S) > e^\epsilon p'(S)\}$ and $\mathcal{S}_2=\{ S \mid p(S) \leq e^\epsilon p'(S)\}$.
Rewriting the above inequality, we have
$$
\begin{array}{rl}
  & \sum_{S\in \mathcal{S}_1} p(S)q(S) + \sum_{S\in \mathcal{S}_2} p(S)q(S)
\\
> & e^\epsilon\sum_{S \in \mathcal{S}_1} p'(S)q(S)+ e^\epsilon\sum_{S \in \mathcal{S}_2} p'(S)q(S) + \delta
\end{array}
$$
Consider the sum over $\mathcal{S}_2$, we have
$$\sum_{S\in \mathcal{S}_2} p(S)q(S) \leq e^\epsilon \sum_{S \in \mathcal{S}_2} p'(S)q(S)$$
Subtracting the above from previous, we have
$$\sum_{S\in \mathcal{S}_1} p(S)q(S) > e^\epsilon\sum_{S \in \mathcal{S}_1} p'(S)q(S)+ \delta$$
For each $S\in \mathcal{S}_1$, we have $p(S)(1-q(S)) > e^\epsilon p'(S)(1-q(S))$, and thus
$$\sum_{S\in \mathcal{S}_1} p(S)(1-q(S)) > e^\epsilon\sum_{S \in \mathcal{S}_1} p'(S)(1-q(S))$$
Summing up the above two inequalities, we have
$$\sum_{S\in \mathcal{S}_1} p(S) > e^\epsilon\sum_{S \in \mathcal{S}_1} p'(S) + \delta$$
This contradicts that $\AA_1$ satisfies \bedDPS.
\end{proof}

\begin{theorem}\label{thm:convexity}
Given two algorithms $\AA_1$ and $\AA_2$ that both satisfy \bedDPS, for any $p\in [0,1]$, let $\AA_p(D)$ be the algorithm that outputs $\AA_1(D)$ with probability $p$ and $\AA_2(D)$ with probability $1-p$, then $\AA_p$ satisfies \bedDPS.
\end{theorem}
\begin{proof}
Since both $\mathcal{A}_1$ and $\mathcal{A}_2$ satisfy \bedDPS, for any pair of neighboring datasets $D$ and $D^{\prime}$ and for any $O\in \mathsf{Range}(\mathcal{A}_1)\cup\mathsf{Range}(\mathcal{A}_2)$, we have
$$
\begin{array}{rl}
 & \Pr[\mathcal{A}_{p}(D)\in O]
\\
=& p\Pr[\mathcal{A}_1(D)\in O]+(1-p)\Pr[\mathcal{A}_2(D)\in O]\\
\leq & p(e^{\epsilon}\Pr[\mathcal{A}_1(D^{\prime})\in O]+\delta)+(1-p)( e^{\epsilon}\Pr[\mathcal{A}_2(D^{\prime})\in O]+\delta)\\
 = & e^{\epsilon}(p\Pr[\mathcal{A}_1(D^{\prime})\in O]+(1-p)\Pr[\mathcal{A}_2(D^{\prime})\in O])+\delta\\
 = & e^{\epsilon}\Pr[\mathcal{A}_{p}(D^{\prime})\in O]+\delta.
\end{array}
$$
Therefore, the algorithm $\mathcal{A}_p$ also satisfies \bedDPS.
\end{proof}

\comment{
\todo{Decide whether to include the following about composability between \edDP and \bedDPS.}

\begin{theorem}
Given $\AA_1$ that satisfies $\DP{(\epsilon_1,\delta_1)}$, and $\AA_2$ that satisfies $\DPS{(\beta,\epsilon_2,\delta_2)}$, then $\AA(D)=(\AA_1(D);\AA_2(D))$ satisfies $\DPS{(\beta,\ln\left(1+\beta\left(e^{\epsilon_1}-1\right)\right)+\epsilon_2,\beta\delta_1+\delta_2)}$.
\end{theorem}
\begin{proof}
For any neighboring $D,D'$, and any $O=O_1;O_2$, $\AA(D)\in O$ when $\AA_1(D)\in O_1$ and $\AA_2(D)\in O_2$. We have
$$
\begin{array}{rl}
 \frac{\Pr[(\AA_1;\AA_2)^\beta(D)\in O]}{\Pr[(\AA_1;\AA_2)^\beta(D')\in O]}= &  \frac{\Pr[(\AA_{1}^\beta(D);\AA_{2}^{\beta}(D))=(S_1;S_2)]}{\Pr[(\AA_{1}^{\beta}(D_{-t});\AA_{2}^{\beta}(D_{-t}))=(S_1;S_2)]}
 \\
 = & \frac{\Pr[\AA_{1}^{\beta}(D)=S_1]}{\Pr[\AA_{1}^{\beta}(D_{-t})=S_1]} \frac{\Pr[\AA_{2}^{\beta}(D)=S_2]}{\Pr[\AA_2^{\beta}(D_{-t})=S_2]}
\end{array}
$$
\end{proof}
}

\subsection{More Non-Composability}

From observations in Section~\ref{sec:dps:uncertain}, we expect that \bedDPS does not compose.  However, one would expect that combining an algorithm that satisfies \DPS{(\beta,\epsilon_1,\delta)} and one that satisfies \DP{\epsilon_2} should result in an algorithm that satisfies the weaker \bedDPS, where $\epsilon$ is some function of $\epsilon_1$ and $\epsilon_2$.  Such a weaker form of composability is useful in that given a dataset that is resulted from random sampling, one can publish it in a way that satisfies \bedDPS, while at the same time answering queries using mechanisms that satisfy \eDP.  Surprisingly, even such a weak form of composability does not hold.

Consider the following two algorithms operating on datasets in which each tuple has two fields: gender and name.  Let $r(D)$ be the predicate that $D$ contains more male than female, and $s(D)$ be a sensitive predicate, such as whether $D$ contains a specific tuple.  The algorithm $\AA_1(D)$ outputs $r(D)$ XOR $s(D)$ when $D$ contains a sufficient number of tuples (say, $1000$), and outputs false otherwise.  And $\AA_2(D)$ outputs the percentage of tuples in $D$ that are male with Laplacian noise~\cite{DMN+06}.

Clearly $\AA_2(D)$ satisfies \eDP.  $\AA_1$ satisfies \bedDPS for any $\beta$ that is not too close to $1$. 
Let $T$ and $T'$ be the random variables resulted from sampling from $D$ and $D'$ respectively. 
Only when the dataset size is large enough, would $\AA_1$ output information that depends on the input data.  When $D$ and $D'$ contain a large number of tuples and differ only by one, $r(T)$ and $r(T')$ have essentially the same distribution, taking the value true with probability very close to $0.5$, making $\AA_1(T)$ and $\AA_1(T')$ having a similar distribution.  
Combining $\AA_1$ and $\AA_2$, however, is non-private.  Using $\AA_2(T)$ one obtains a highly accurate estimate of the predicate $r(T)$, enabling the adversary to learn $s(T)$ with high probability.  

More specifically, let $D$ and $D'$ be two datasets such that $s(D)$ is false, $s(D')$ is true (i.e., $D'$ contains the tuple we are checking), and they each contain 10,000 tuples, half male and half female.  Consider sampling probability $\beta=0.5$, and the event that $\AA_1$ outputs false, and $\AA_2$ outputs $p\geq 0.5$.  Let $T$ and $T'$ be the random variables resulted from sampling from $D$ and $D'$ respectively, then we have
$$
\begin{array}{rl}
\Pr[s(T)=\mbox{true}] &=0
\\
 \Pr[s(T')=\mbox{true}] & =1/2
 \\  
\Pr[r(T)=\mbox{true} \mid \AA_2(T)\geq 0.5] & \approx 1,
\\
\Pr[r(T')=\mbox{true} \mid \AA_2(T') \geq 0.5] & \approx 1
\end{array}
$$
and
$$
\begin{array}{rl}
 & \Pr[\AA_2(T)\geq 0.5 \wedge \AA_1(T)=\mbox{false}]
 \\
= & \Pr[\AA_2(T)\geq 0.5] \Pr[ r(T)=s(T) \mid \AA_2(T)\geq 0.5]
 \\
\approx & \Pr[\AA_2(T)\geq 0.5] \Pr[ s(T) = \mbox{true}]
 \\
= &  0,
\end{array}
$$
while
$$
\begin{array}{rl}
&\Pr[\AA_2(T')\geq 0.5 \wedge \AA_1(T')=\mbox{false}]
\\
\approx & \Pr[\AA_2(T')\geq 0.5] \Pr[ s(T')=\mbox{true} \mid \AA_2(T')\geq 0.5]
\\
= & \Pr[\AA_2(T')\geq 0.5] \Pr[ s(T')=\mbox{true}]
\\
\approx & 1/4.
\end{array}
$$

This result is somewhat surprising.  After all, any mechanism that satisfies \eDP should not be leaking private information about the underlying datasets.  How could adding a differentially private mechanism destroys the privacy protection of another mechanism?
Our understanding is that satisfying \bedDPS can be achieved by relying on the adversary's uncertainty.  The adversary knows only that the dataset is from a large set of candidates.  While \eDP ensures that adjacent datasets are difficult to distinguish, these candidates are not all adjacent and can indeed be quite far apart.  Hence obtaining one \eDP answer may dramatically change the probability of which candidates are possible, removing some degree of uncertainty, destroying any privacy protection that relies on exactly that uncertainty.

This inability for a \bedDPS mechanism to compose with a \eDP mechanism suggests that \bedDPS mechanisms should be applied alone.  Hence they are not suitable for the interactive mode, but only suitable for the non-interactive mode of data publishing. 
Furthermore, it also suggests that mechanisms satisfying \eDP should be used carefully as well, as its output may break other mechanisms' (albeit weaker) privacy guarantees.

\subsection{Benefiting from Sampling}

We observe that in many data publishing scenarios, random sampling is an inherent step.  For example, the census bureau publishes a 1-percent microdata sample.  In many research settings (such as when Netflix wants to publishing movie ratings), it is sufficient to publish a random sample of the dataset.  Many times, even when the dataset is not the result of explicit sampling, one can view it as result of implicit sampling, because the process of selecting respondents involves randomness.



The natural question is how one can benefit from such explicit or implicit sampling.
Our results provide the following answers.  The first way is to limit oneself to mechanisms that satisfy $\edDP$, then the uncertainty resulted from sampling enables one to use a larger privacy budget because of the amplification result in Theorem~\ref{thm:amplification}.  The second way is to use a mechanism that does not satisfy $\edDP$, but satisfies $\bedDPS$, such as safe $k$-anonymization, which we will study in Section~\ref{sec:result}.  However, this way of benefiting from sampling can be enjoyed only once; one cannot use the same dataset to answer other queries, even when using mechanisms that satisfy \eDP.



There, however, does exist a more flexible way to use a mechanism that satisfies only \bedDPS.  When one has a large dataset, one can sample a dataset, apply the mechanism, publish the result, and discard the intermediate sampled dataset.  Because of the composability of \edDP, this approach can be applied multiple times so long as each time one performs a fresh sampling.  One can also use multiple mechanisms that satisfy \edDP on a newly sampled dataset.

We point out that the benefit of sampling should not be viewed as just ``throwing away data''; sampling's main benefit is to introduce uncertainty.  Given a dataset, one could sample with, say, $\beta= 0.2$ for many (say, 50) times, and apply a mechanism that satisfies \DPS{(0.2, 0.02, 0)} to each sampled dataset and publish the results.  With high probability, each tuple is included in at least one of the sampled datasets.  That is, in some sense, no tuple is thrown away.  However, as each sampling and publishing satisfies \edDP, and \edDP composes, publishing the 50 outputs still satisfies \edDP for $\epsilon=1,\delta=0$.

In summary, sampling creates uncertainty for the adversary.  While the benefit due to this uncertainty is easy to lose because the uncertainty can be jeopardized by answering any query on it, this uncertainty is also easy to gain, as each sampling introduces fresh uncertainty.

\section{Safe $k$-Anonymization Meets Differential Privacy}\label{sec:result}

In this section we show that $k$-anonymization, when performed in a ``safe'' way, satisfies \bedDPS.  That is, safe $k$-anonymization, when preceded by a random sampling step, satisfies $(\epsilon,\delta)$-differential privacy.

\subsection{An Analysis of $k$-Anonymity}~\label{sec:kanon}




The development of $k$-anonymity was motivated by a well publicized privacy incident~\cite{Swe02a}. The
Group Insurance Commission (GIC) published a supposedly anonymized dataset recording the medical visits of patients managed under its insurance plan.  While the obvious personal identifiers (such as name and address) were removed, the published data included zip code, date of birth, and gender, which are sufficient to uniquely identify a significant fraction of the population.  Sweeney~\cite{Swe02a} showed that by correlating this data with the publicly available Voter Registration List for Cambridge Massachusetts, medical visits for many individuals can be easily identified, including those of William Weld, a former governor of Massachusetts.  We note that even without access to the public voter registration list, the same privacy breaches can occur.  Many individuals' birthdate, gender and zip code are public information.
This is especially the case with the advent of social media, including Facebook, where users share seemingly innocuous personal information to the public.  The GIC re-identification attack directly motivated the development of the $k$-anonymity privacy notion.

\begin{definition} \textbf{[$k$-Anonymity, the privacy notion]~\cite{Swe02a}:} \label{def:k-anonymity}
A published table satisfies $k$-anonymity relative to a set of QID attributes if and only if when the table is projected to include only the QIDs, every tuple appears at least $k$ times.
\end{definition}

\mypara{Quasi-identifiers vs. Sensitive Attributes?}
A first problem with Definition~\ref{def:k-anonymity} is that it requires the division of all attributes into quasi-identifiers (QIDs) and sensitive attributes (SA), where the adversary is assumed to know the QIDs, but not SAs.  This separation, however, is very hard to obtain in practice.  Even though only some attributes are used in the GIC incident, it is difficult to assume that they are the only QIDs.  Other attributes in the GIC data include visit date, diagnosis, etc.  There may well exist an adversary who knows this information about some individuals, and if with this knowledge these individuals' record can be re-identified, it is still a serious privacy breach.

The same difficulty is true for publishing any kind of census, medical, or transactional data.  When publishing anonymized microdata, one has to defend against all kinds of adversaries, some know one set of attributes, and others know different sets.  An attribute about one individual may be known by some adversaries, and unknown (and should be considered sensitive) for other adversaries.
\comment{
It seems that the only reasonable assumption one can make is the following:
\begin{quote}\textbf{[Complete Individual Knowledge Assumption]}
For every individual, there may exist an adversary who knows \emph{all attributes} of that individual.
\end{quote}
}

Any separation between QIDs and SAs is essentially making assumptions about the adversary's background knowledge that can be easily violated, rendering any privacy protection invalid.  Hence we consider a strengthened version of $k$-anonymity by treating all attributes as QIDs.  This is stronger than using any subset of attributes as QIDs.  This strengthened version of $k$-anonymity avoids making assumption about the adversary's background knowledge about which attributes are known and what are not.  This has been used in the context of anonymizing transaction data \cite{HN09}.

\mypara{Weakness of the $k$-Anonymity Notion.}  With the strengthened version of $k$-anonymity, one might expect that it should stop re-identification attacks.  To satisfy this notion, each tuple in the output is blended in a group of at least $k$ tuples that are the same.  This follows the appealing principle that ``privacy means hiding in a crowd''.  The intuition is that as there are at least $k-1$ other tuples that look exactly the same, one cannot re-identify which tuple in the output corresponds to an individual with probability over $1/k$.  Unfortunately, this intuition turns out to be wrong.  Only making the syntactic requirement that each tuple appears at least $k$ times does not protect privacy, as a trivial way to satisfy this is to select some tuples from the input and then duplicate each of them $k$ times.



Several other privacy notions have been introduced on the motivation that $k$-anonymity is not strong enough.  Among these are  $\ell$-diversity~\cite{MGK+06} and $t$-closeness~\cite{LLV07}.  In these approaches, it is observed that even if $k$-anonymity is achieved, information about sensitive attributes can still be learned, perhaps due to the uneven distribution of their values.  This line of work, however, still requires the problematic assumption that there is a separation between QIDs and SAs, and that the adversary knows only the QIDs.  In other words, while they correctly assert that $k$-anonymity is not strong enough, these definitions did not fix it in the right way.


\mypara{$k$-Anonymity vs.~$k$-Anonymization Algorithms.}
Here we would like to make a clear distinction between the \emph{$k$-anonymity}, the privacy notion, and \emph{$k$-anonymization algorithms}.

Many $k$-anonymization algorithms have been developed in the literature.  Given input datasets, they aim at producing anonymized versions of the input datasets that satisfy $k$-anonymity.  That the $k$-anonymity privacy notion is weak means that producing outputs that \emph{satisfying $k$-anonymity alone} is insufficient for privacy protection.  However, this does not automatically mean that all $k$-anonymization have privacy vulnerabilities.  We now show that the algorithms that have been developed in the literature are indeed vulnerable to re-identification attacks.  Consider the following anonymization scheme, which represents several proposed algorithms for $k$-anonymity~\cite{BKBL07,LDR06}.

\begin{alg}\textbf{[Clustering and Local Recoding (CLR)]:}
First, group input tuples into clusters such that each cluster has at least $k$ tuples.  For example, one method of grouping is the Mondrian algorithm~\cite{LDR06}.  One could also use some clustering method based on some distance measurement (e.g., \cite{BKBL07}).  Then, for each tuple, replace each attribute value with a generalized value that represents all values for that attribute in the cluster.
\end{alg}

CLR algorithms are vulnerable when some tuples contain extreme values.  Even if the output satisfies $k$-anonymity, the generalized value depends on the extreme values of some tuples; hence from the output an adversary can infer that one's tuple is in the dataset and can thus infer these values.  For example, suppose the dataset records the net worth of some individuals in a town.  Further suppose that it is known that only one individual in the town has net worth over \$10 million.  When given a $(k=20)$-anonymized output dataset containing one group of tuples that all have $[900K,35M]$ as the generalized net worth value, what can one conclude?  At least the following: the rich individual is in the dataset; the individual's tuple is in the group; and the individual's net worth is \$35 million.  It would be difficult to say that because in the output dataset, there are at least $19$ other tuples that are exactly the same, then the individual cannot be re-identified with probability $1/20$.


Similar weaknesses exist for other $k$-anonymization algorithm in the literature, for example, those computing a generalization scheme based on the input dataset~\cite{HN09}.  With all these algorithms, the presence and non-presence of some extreme values will affect the resulted generalization scheme, leaking information.

As these algorithms are sensitive to the presence of a single tuple with extreme values, they do not satisfy \bedDPS when $\beta > \delta$, since sampling with $\beta$ will result the presence of the tuple selected with probability $\beta$.


\subsection{Towards ``Safe'' $k$-Anonymization}

We have shown that $k$-anonymity (even when all attributes are treated as QID) does not provide adequate protection, nor do existing $k$-anonymization algorithms.  One natural question is: Is this because the intuition ``hidden in a crowd'' fails to provide privacy protection, or is it because the definition of $k$-anonymity fails to correctly capture ``hidden in a crowd''?

We believe that the answer is the latter.  The notion of $k$-anonymity implicitly assumes that there is a one-to-one relation $g$ between the input tuples and the output tuples, i.e., given input $D$, the output dataset is $\{g(t) \mid t\in D\}$.  When there are $k$ output tuples that are the same, there must exist $k$ input tuples that are indistinguishable based only on their corresponding outputs.  However, this relation $g$ itself can be overly dependent on one or a few input tuples.  For instance, consider the example above with the extreme value.  Choosing $[900K,35M]$ as the generalized value depends on the single input tuple with value $35M$; hence all tuples that contain this generalized value are directly affected by one tuple's presence, and the tuple is not really ``hiding in a crowd''.


An intriguing question is: If a $k$-anonymization algorithm uses a mapping that does not overly depend on any individual tuple, does such an algorithm provide an adequate level of privacy protection?  To answer this question, we first formalize such algorithms as safe $k$-Anonymization algorithms.

Intuitively, an $k$-anonymization algorithm $\AA$ takes as input a dataset $D$ and a value $k$ and produces an output dataset $S=\AA(D)$.  In order to define ``safe'' anonymization algorithms, we require each anonymization algorithm $\AA$ to be specified in two steps.  The first step, $\AA_m$, outputs a mapping function $g:D\rightarrow T$, where $T$ is the set of all possible tuples. The second step applies $g$ to all tuples in $D$.  That is, $\AA(D, k)=\mathsf{Apply}(\AA_m(D, k), D, k)$, where $\mathsf{Apply}$ is defined as follows.


\begin{algorithm}
	$\mathsf{Apply}(g, D, k)$
    \begin{algorithmic}
        \STATE $S \leftarrow \emptyset$
        \FORALL{$t \in D$}
            \STATE $S \leftarrow S \cup g(t)$
        \ENDFOR
        \FORALL{$s \in S$}
            \IF{$s\mbox{ appears less than $k$ times in }S$}
                \STATE remove all occurrences of $s$ from $S$
            \ENDIF
        \ENDFOR
		\RETURN $S$
	\end{algorithmic}
\end{algorithm}


We note that all existing $k$-anonymization algorithms can be modeled this way, as there is no limitation on the the form of $\AA_m$'s output $g$.  In the extreme case, $g$ can be described as a table matching each tuple in $D$ to the desired output tuple.

\begin{definition}[Strongly-Safe Anonymization]
We say that a $k$-anonymization algorithm \AA is strongly safe if and only if the function $\AA_m(D,k)$ is remains constant when $D$ changes, i.e., the mapping $g$ does not depend on its input dataset.
\end{definition}

An example of a strongly-safe $k$-anonymization algorithm is to always use the same global recoding scheme no matter what dataset is the input.


Intuitively a strongly-safe $k$-anonymization algorithm provides some level of privacy protection, and the level of privacy protection increases with larger values of $k$.  If any individual's tuple is published, there must exist at least $k-1$ other tuples in the \emph{input} database that are the same under the recoding scheme; furthermore, the recoding scheme does not depend on the dataset, and one sees only the results of the recoding.  Hence in this input dataset, the individual is hidden in a crowd of at least $k$.  However, the following proposition shows that strongly safe $k$-anonymization algorithms do not satisfy \edDP.

\begin{prop}
No strongly-safe $k$-anonymization algorithm satisfies \edDP for any $\delta<1$.
\end{prop}
\begin{proof}
Given a strongly-safe algorithm $\AA$, let $g$ be the mapping \AA uses.  Choose $D$ and $D'$ that differ in one tuple $t$ and $D$ contains $n>k$ tuples $t'$ such that $g(t')=g(t)$. The dataset $D'$ contains $n-1$ copies of such $t'$.  Then, $\AA(D)$ and $\AA(D')$ contain different numbers of $g(t)$.  Let $S=\AA(D)$, we have $\Pr[\AA(D)=S]=1$ and $\Pr[\AA(D')=S]=0$.
\end{proof}


\comment{
We note that an adversary is able to tell whether an output $S$ is resulted from $D$ or $D'$ because the adversary knows exactly how many copies of $t'$ are in $D$ such that $g(t')=g(t)$.  A natural consideration is: ``Can we make the adversary's knowledge about this less certain?''  Observing an output that contains, say, $25$ copies of tuples $g(t)$, an adversary would be unable to tell whether the input is $D$ or $D'$ if he is uncertain whether $D$ contains 25 or 26 copies of $t'$ such that $g(t')=g(t)$.
}

\subsection{Privacy of Strongly-Safe $k$-Anonymization}


We now show that strongly-safe $k$-anonymization algorithm satisfies $(\beta, \epsilon, \delta)$-differential privacy for a small $\delta$ with reasonable values of $k$ and $\beta$.  We use $f(j;n,\beta)$ to denote the probability mass function for the binomial distribution; that is, $f(j;n,\beta)$ gives the probability of getting exactly $j$ successes in $n$ trials where each trial succeeds with probability $\beta$. And we use $F(j;n,\beta)$ to denote the cumulative probability mass function; that is, $F(j;n,\beta)=\sum_{i=0}^{j} f(i;n,\beta)$.


\comment{
\begin{definition}[$(\epsilon,\delta)$-differential privacy (\edDP)]\label{def:pwindist}
A randomized algorithm $\AA$ satisfies $(\epsilon,\delta)$-differential privacy, if for
any dataset $D$ and any $t \in D$, the
following holds with probability at least $1-\delta$:
$$
e^{-\epsilon}   \leq \frac{\Pr[\AA(D)=S]}{\Pr[\AA(D_{-t})=S]} \leq e^{\epsilon}
$$
More precisely, let $O$ denote the set of all $S$'s in $\mathit{Range}(\AA)$ such that the above does not hold.  Then $\Pr[\AA(D)\in O] < \delta$ and $\Pr[\AA(D_{-t})\in O] < \delta$.
\end{definition}

Some papers use another slightly weaker version of relaxation, which is implied by Definition~\ref{def:pwindist}. See, for example,~\cite{KS08} for more details.
}

\begin{theorem}\label{thm:main}
Any strongly-safe $k$-anonymization algorithm satisfies \bedDPS for any $0<\beta<1$, $\epsilon\geq - \ln(1-\beta)$, and $\delta = d(k,\beta,\epsilon)$, where the function $d$ is defined as
\begin{eqnarray*}
d(k,\beta,\epsilon) &= &
\max_{n: n \geq \left\lceil \frac{k}{\gamma} -1 \right\rceil }\sum_{j>\gamma n}^{n} f(j; n, \beta),
\\
\end{eqnarray*}
where $\gamma = \frac{(e^\epsilon-1+\beta)}{e^\epsilon}$.
\end{theorem}
See Appendix~\ref{app:main} for the proof.

\begin{table*}
	\begin{center}
		$\epsilon$ \\ \ \\
		$\beta$
		\begin{tabular}{c||cccccc}
		& 0.25 & 0.5 & 0.75 & 1.0 & 1.5 & 2.0 \\
		 \hline
		0.05 & 6.83$\times 10^{-10}$ & 2.50$\times 10^{-14}$ & 3.19$\times 10^{-17}$ & 1.76$\times 10^{-19}$ & 3.97$\times 10^{-22}$ & 2.00$\times 10^{-24}$ \\
		 \hline
		0.1 & 4.19$\times 10^{-06}$ & 1.61$\times 10^{-09}$ & 3.44$\times 10^{-12}$ & 4.07$\times 10^{-14}$ & 3.22$\times 10^{-16}$& 1.89$\times 10^{-18}$ \\
		 \hline
		0.2 & 2.16$\times 10^{-03}$ & 8.02$\times 10^{-06}$ & 1.89$\times 10^{-07}$ & 6.03$\times 10^{-09}$ & 4.79$\times 10^{-11}$ & 1.59$\times 10^{-12}$ \\
		 \hline
		\end{tabular}
		\caption{A table showing the relationship between $\beta$ and $\epsilon$ in determining the value of $\delta$ when $k$ is fixed. In the above $k=20$, and each cell in the table reports the value of $\delta$ under the given values of $\beta$ and $\epsilon$}\label{tab:betaeps}
	\end{center}
\end{table*}

The function $d$ relates the four parameters $\epsilon,\beta,k,\delta$ by requiring $\delta=d(k,\beta,\epsilon)$.  Note that the other requirement is that $\epsilon\geq - \ln(1-\beta)$.  Among the four parameters, $\epsilon$ and $\delta$ define the level of privacy protection, while $k$ and $\beta$ affect the quality of anonymized data.  We now examine the relationships among these four parameters.


To compute this, we want to find $n \geq \left\lceil \frac{k}{\gamma} -1 \right\rceil $ that maximizes $\sum_{j>\gamma n}^{n} f(j; n, \beta)$.  We first observe that $\gamma > \beta$ because
$$
\begin{array}{ll}
\gamma - \beta = \frac{(e^\epsilon-1+\beta)}{e^\epsilon} - \beta = \frac{(e^\epsilon-1)(1-\beta)}{e^{\epsilon}} > 0
\end{array}
$$

That is, $\sum_{j>\gamma n}^{n} f(j; n, \beta)$ sums up the tail binomial distribution probabilities for the portion of the tail beyond $\gamma n$, as shown in Figure~\ref{fig:binom}.  Following the intuition behind the law of large numbers, the larger the value of $n$, the smaller this tail probability.  
Hence intuitively, choosing the smallest value of $n$, i.e., $n=n_m=\left\lceil \frac{k}{\gamma} -1 \right\rceil$ should maximize the formula.
Unfortunately, due to the discrete nature of the binomial distribution, the maximum value may not be reached at $n_m$, but instead at one of the next few local maximal points $\left\lceil \frac{k+1}{\gamma} -1 \right\rceil$, $\left\lceil \frac{k+2}{\gamma} -1 \right\rceil$, $\cdots$.
Thus we are unable to further simplify the representation of the function $d(k,\beta,\epsilon)$.

We now report the relationships among $\epsilon,\beta,k,\delta$ using numerical computation.
In Table~\ref{tab:betaeps}, we fix $k=20$ and report the values of $\delta$ under different $\epsilon$ and $\beta$ values. The table shows that the values of $\delta$ can be very small. We note that with fixed $k$ and $\beta$, $\delta$ decreases as $\epsilon$ increases, which states that the error probability gets smaller when one relaxes the $\epsilon$-bound on the probability ratio.  In other words, the more serious a privacy breach, the more unlikely it occurs.  The table also shows that with fixed $k$ and $\epsilon$, $\delta$ decreases as $\beta$ decreases, meaning that a smaller sampling probability improves the privacy protection.


In Figure~\ref{fig:eps_delta1}, we show the results from examining the relationship between $\epsilon$ and $\delta$ when we vary $k\in \{5,10,20,30,50\}$ under fixed $\beta=0.2$.  We plot $\frac{1}{\delta}$ against $\epsilon$ for values of $\epsilon > -\ln(1-\beta)$. The figure indicates a negative correlation between $\epsilon$ and $\delta$. Furthermore, increasing $k$ has a close to exponential effect of improving privacy protection.  For example, when $\epsilon=2$, increasing $k$ by $10$ roughly decreases $\delta$ by $10^{-5}$.

In Figure~\ref{fig:eps_delta2}, we show the results from examining the effect of varying $\beta \in \{0.05,0.1,0.2,0.3,0.4\}$ under a fixed value of $k=20$.  This shows that decreasing $\beta$ also dramatically improve the privacy protection.  The two figures indicate the intricate relationship between privacy and utility.

In Figure~\ref{fig:eps_delta3}, we explore this phenomenon that increasing $k$ and decreasing $\beta$ both improve privacy protection.  Starting from $(k=15,\beta=0.05)$, each time we double $\beta$ and find a value $k$ that gives a similar level of privacy protection.  We finds that $k$ increases from $15$ to $22$ (for $\beta=0.1$), $35$ (for $\beta=0.2$), and $60$ (for $\beta=0.4$).

In Figure~\ref{fig:eps_delta4}, we examine the quality of privacy protection for very small $k$'s (from $1$ to $5$).  We choose a very small sampling probability of $\beta=0.025$.  Not surprisingly, when $k=1$, the privacy protection is entirely from the sampling effect, as the obtained $\delta$ value is less than $\beta$.  However, when $k\geq 2$, we start seeing privacy protection effect from $k$-anonymization, with $\delta$  ($< 0.001$)  significantly smaller than $\beta=0.025$ when $\epsilon=2$.

Finally, in Figure~\ref{fig:k_eps} we show the relationship between the privacy parameter $\epsilon$ and the utility parameter $k$ if we set the requirement that $\delta \leq 10^{-6}$. The figure shows that smaller values of $\epsilon$ can be satisfied for larger values of k. Furthermore, the effect of $\beta$ over $\epsilon$ is quite substantial.

\begin{figure}
	\includegraphics[width=3.5in, height=2in]{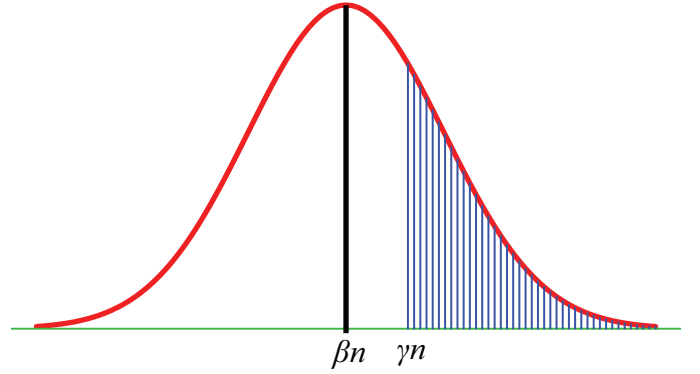}
	\caption{A graph showing the relationship between $\beta n$ and $\gamma n$ on a binomial curve}\label{fig:binom}
\end{figure}

\begin{figure}
	\psfrag{DELTA}{$\frac{1}{\delta}$}
	\psfrag{EEPS}{$\epsilon$}
	\psfrag{KK = 50, BB = 0.2}{\tiny $k = 50, \beta = 0.2$}
	\psfrag{KK = 30, BB = 0.2}{\tiny $k = 30, \beta = 0.2$}
	\psfrag{KK = 20, BB = 0.2}{\tiny $k = 20, \beta = 0.2$}
	\psfrag{KK = 10, BB = 0.2}{\tiny $k = 10, \beta = 0.2$}
	\psfrag{KK = 5, BB = 0.2}{\tiny $k = 5, \beta = 0.2$}
	\includegraphics[width=3.5in]{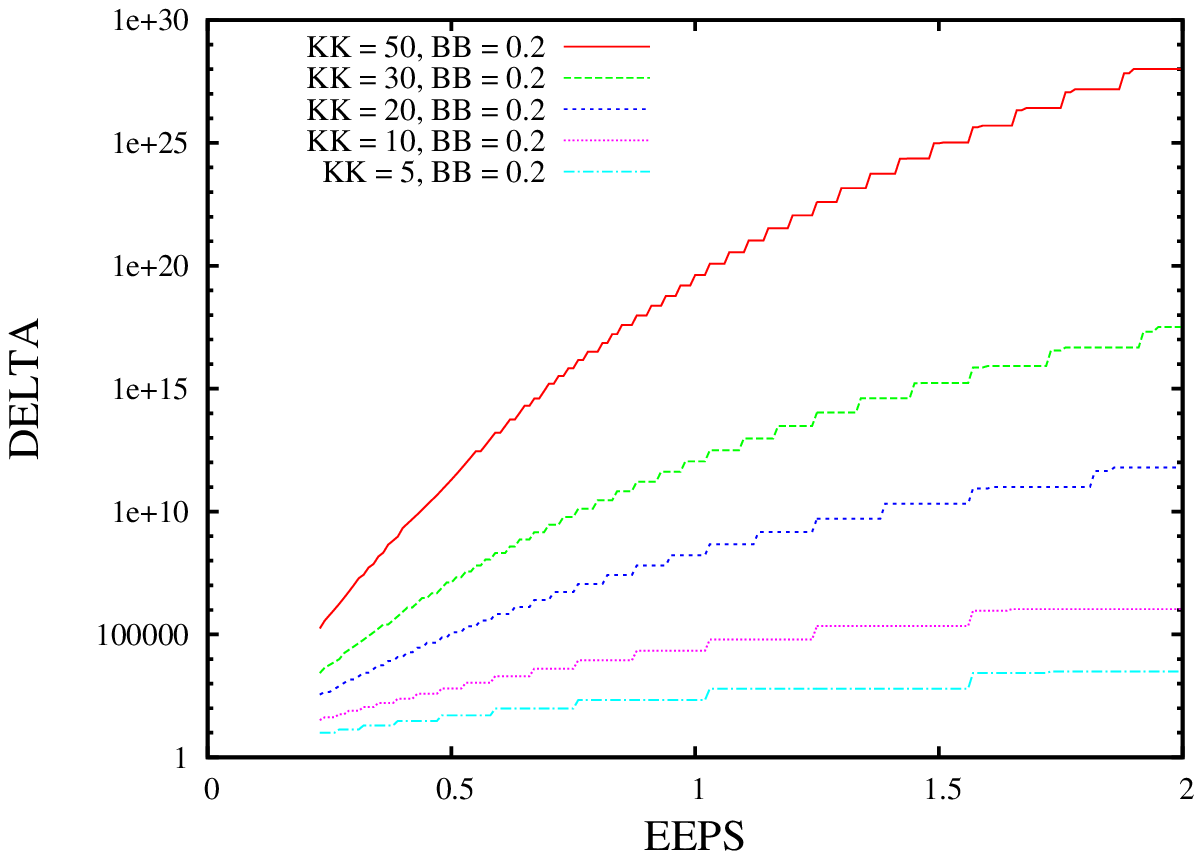}
	\caption{A graph showing the relationship between $\epsilon$ and $\frac{1}{\delta}$ if we vary the values of $k$ under fixed $\beta$}\label{fig:eps_delta1}
\end{figure}

\begin{figure}
	\psfrag{DELTA}{$\frac{1}{\delta}$}
	\psfrag{EEPS}{$\epsilon$}
	\psfrag{KK = 20, BB = 0.05}{\tiny $k = 20, \beta = 0.05$}
	\psfrag{KK = 20, BB = 0.1}{\tiny $k = 20, \beta = 0.1$}
	\psfrag{KK = 20, BB = 0.2}{\tiny $k = 20, \beta = 0.2$}
	\psfrag{KK = 20, BB = 0.3}{\tiny $k = 20, \beta = 0.3$}
	\psfrag{KK = 20, BB = 0.4}{\tiny $k = 20, \beta = 0.4$}
	\includegraphics[width=3.5in]{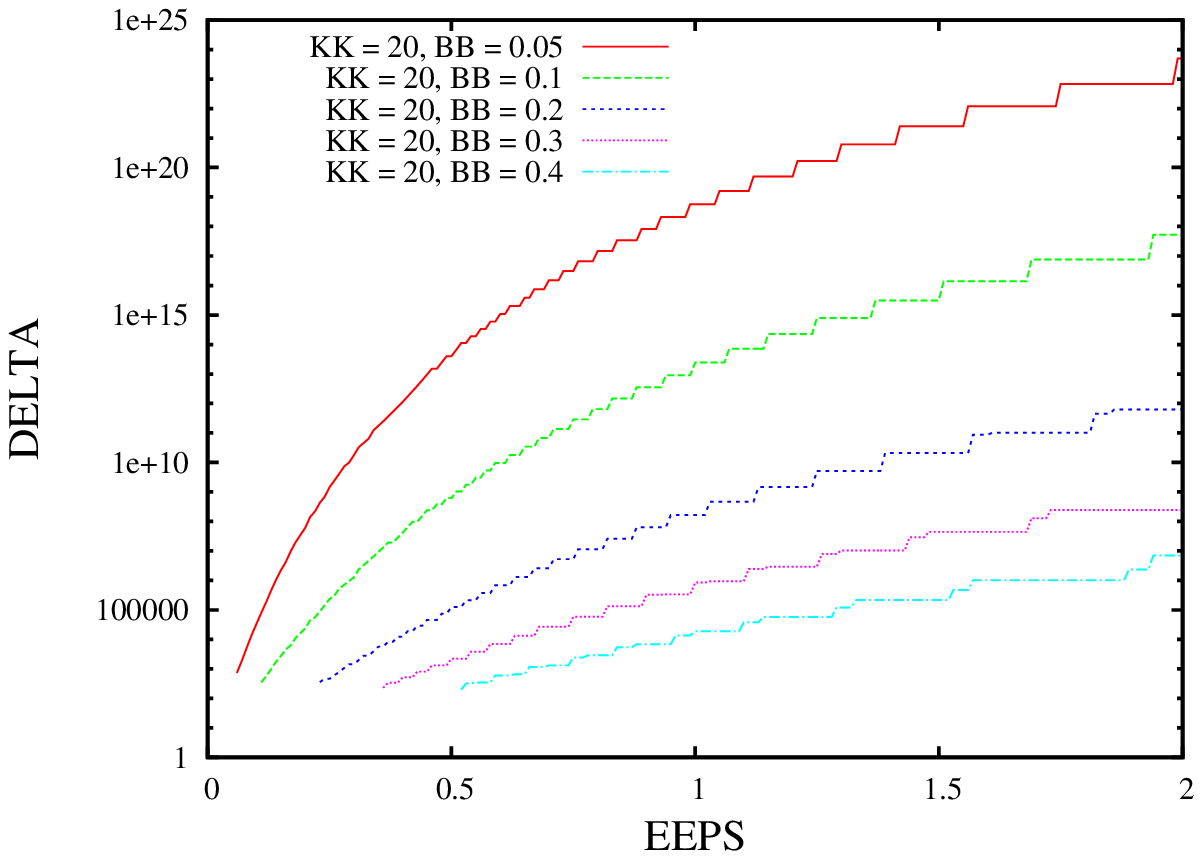}
	\caption{A graph showing the relationship between $\epsilon$ and $\frac{1}{\delta}$ if we vary the values of $\beta$ under fixed $k$.}\label{fig:eps_delta2}
\end{figure}

\begin{figure}
	\psfrag{DELTA}{$\frac{1}{\delta}$}
	\psfrag{EEPS}{$\epsilon$}
	\psfrag{KK = 15, BB = 0.05}{\tiny $k = 15, \beta = 0.05$}
	\psfrag{KK = 22, BB = 0.1}{\tiny $k = 22, \beta = 0.1$}
	\psfrag{KK = 35, BB = 0.2}{\tiny $k = 35, \beta = 0.2$}
	\psfrag{KK = 60, BB = 0.4}{\tiny $k = 60, \beta = 0.4$}
	\includegraphics[width=3.5in]{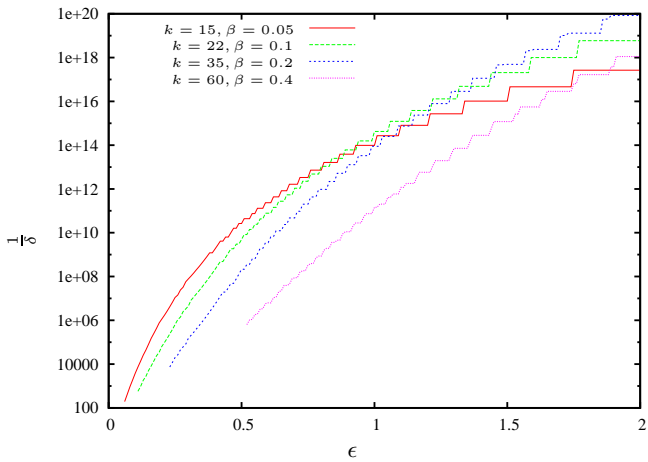}
	\caption{A graph showing the relationship between the values of $k$ needed to achieve roughly the same $\delta$ if we double $\beta$.}\label{fig:eps_delta3}
\end{figure}

\begin{figure}
	\psfrag{DELTA}{$\frac{1}{\delta}$}
	\psfrag{EEPS}{$\epsilon$}
	\psfrag{KK = 1, BB = 0.025}{\tiny $k = 1, \beta = 0.025$}
	\psfrag{KK = 2, BB = 0.025}{\tiny $k = 2, \beta = 0.025$}
	\psfrag{KK = 3, BB = 0.025}{\tiny $k = 3, \beta = 0.025$}
	\psfrag{KK = 4, BB = 0.025}{\tiny $k = 4, \beta = 0.025$}
	\psfrag{KK = 5, BB = 0.025}{\tiny $k = 5, \beta = 0.025$}
	\includegraphics[width=3.5in]{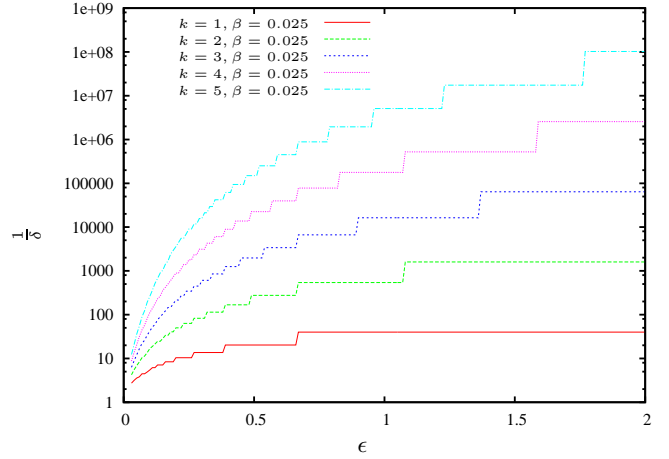}
	\caption{A graph showing the relationship between $\epsilon$ and $\frac{1}{\delta}$ with small $k$'s, varying $k$ and fixing $\beta$.}\label{fig:eps_delta4}
\end{figure}

\begin{figure}
	\psfrag{KK}{$k$}
	\psfrag{EPS}{$\epsilon$}
	\psfrag{BB = 0.05}{\tiny $\beta = 0.05$}
	\psfrag{BB = 0.1}{\tiny $\beta = 0.1$}
	\psfrag{BB = 0.2}{\tiny $\beta = 0.2$}
	\includegraphics[width=3.5in]{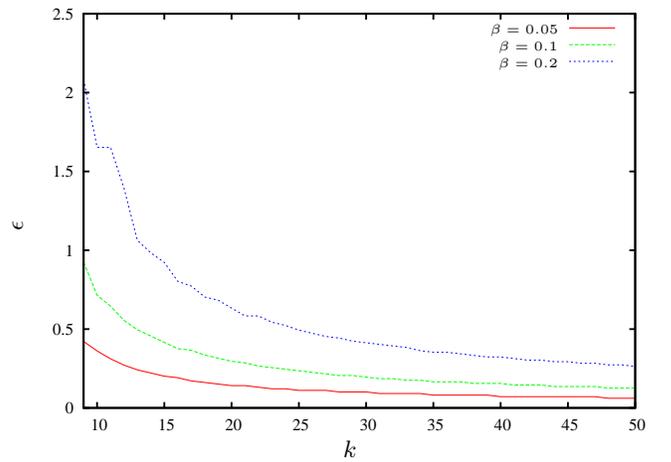}
	\caption{A graph showing the value of $\epsilon$ satisfied by a given $k$ if $\delta \leq 10^{-6}$ with varying sampling probabilities.}\label{fig:k_eps}
\end{figure}




\subsection{$\epsilon$-Safe $k$-Anonymization}

In practice, requiring a $k$-anonymization algorithm to be strongly safe is likely to result in outputs that have low utility.  We now relax this requirement to allow the generalization scheme to be chosen in a way that depends on the input dataset, but does not overly depend on any individual tuple.

\begin{definition}[$\epsilon$-Safe $k$-Anonymization]
We say that a $k$-anonymization algorithm \AA is $\epsilon$-safe if and only if the function $\AA_m$ satisfies \eDP.
\end{definition}

One possible approach to do this is to consider various possible generalization schemes, uses a quality function to assign a quality to each of them, and then uses the exponential mechanism~\cite{MT07} to select in a differentially private way a generalization scheme that gives good utility.

The following theorem shows that $\epsilon$-safe $k$-Anonymization also satisfies \bedDPS.
\begin{theorem}\label{thm:esafe}
Any $\epsilon_1$-safe $k$-anonymization algorithm satisfies \bedDPS, where $\epsilon \geq -\ln(1-\beta)+\epsilon_1$, $\delta = d(k, \beta, \epsilon-\epsilon_1)=\max\limits_{n:n\geq\lceil \frac{k}{\gamma}-1\rceil }\sum_{j>\gamma n}^{n}f(j;n,\beta)$, $\gamma = \frac{(e^{\epsilon-\epsilon_1}-1+\beta)}{e^{\epsilon-\epsilon_1}}$.
\end{theorem}

See Appendix~\ref{app:esafe} for the proof.

\subsection{Remarks of the Result}

Theorems~\ref{thm:main} and~\ref{thm:esafe} show that $k$-anonymization, when done safely, and when preceded by a random sampling step, can satisfy \edDP with reasonable parameters.  In the literature, $k$-anonymization and differential privacy have been viewed as very different privacy guarantees.  $k$-anonymization achieves weak syntactic privacy, and differential privacy provides strong semantic privacy guarantees.  Our result is, to our knowledge, the first to link $k$-anonymization with differential privacy.  This suggests that the ``hiding in a crowd of $k$'' privacy principle indeed offers some privacy guarantees when used correctly.  We note that this principle is used widely in contexts other than privacy-preserving publishing of relational data, including location privacy and publishing of social network data, network packets, and other types of data.

We also observe that another way to interpret our result is that this provides a new method of satisfying \edDP.  Existing methods for satisfying differential privacy include adding noise according to the global sensitivity~\cite{Dwo06}, adding noise according to the smooth local sensitivity~\cite{NRS07}, and the exponential mechanism~\cite{MT07} which directly assigns probabilities to each possible answer in the range.
%
%
%
Our result suggests an alternative approach:  Rather than adding noises to the output, one can add a random sampling step in the beginning and prune results that are too sensitive to changes of a single individual tuple (i.e., tuples that violate $k$-anonymity).  In other words, when the dataset is resulted from random sampling, then one can answer count queries \emph{accurately} provided that the result is large enough.
%
%
An intriguing question is whether other input perturbation techniques can be used to satisfy differential privacy as well.

\section{Related Work} \label{sec:related}

A lot of work on privacy-preserving data publishing considers privacy notions that are weaker than differential privacy.  These approaches typically assume an adversary that knows only some aspects of the dataset (background knowledge) and tries to prevent it from learning some other aspects.  One can always attack such a privacy notion by changing either what the adversary already knows, or changing what the adversary tries to learn.  The most prominent among these notions is $k$-anonymity~\cite{Swe02a, Swe02b}.  Some follow-up notions include $l$-diversity~\cite{MGK+06} and $t$-closeness~\cite{LLV07}.  In this paper, we analyze the weaknesses of $k$-anonymity in detail, and argue that a separation between QIDs and sensitive attributes are difficult to obtain in practice, challenging the foundation of privacy notions such as $l$-diversity, $t$-closeness, and other ones centered on attribute disclosure prevention.


The notion of differential privacy was developed in a series of works \cite{DN03,DN04,BDMN05,DMN+06,Dwo06}.  It represents a major breakthrough in privacy-preserving data analysis.  In an attempt to make differential privacy more amenable to more sensitive queries, several relaxations have been developed, including $(\epsilon,\delta)$-differential privacy \cite{DN03, DN04, BDMN05, DMN+06}.
Three basic general approaches to achieve differential privacy are adding Laplace noise proportional to the query's global sensitivity~\cite{Dwo06,DMN+06}, adding noise related to the smooth bound of the query's local sensitivity~\cite{NRS07}, and the exponential mechanism to select a result among all possible results~\cite{MT07}.
A survey on these results can be found in~\cite{Dwo08}.
Our approach suggests an alterative by using input perturbation rather than output perturbation to add uncertainty to the adversary's knowledge of the data.



Random sampling~\cite{AST05,AH05} has been studied as a method for privacy preserving data mining, where privacy notions other than differential privacy were used.
%
The relationship between sampling and differential privacy has been explored before.  Chauduri and Mishra~\cite{CM06} studied the privacy effect of sampling, and showed a linear relationship between the sampling probability and the error probability $\delta$.
Their result suggests an approach to perform first $k$-anonymization and then sampling as the \emph{last} step.  We instead consider the approach of perform sampling as the \emph{first} step and then $k$-anonymization.  Our result suggests that the latter approach benefits much more from the sampling.



There exists some work on publishing microdata while satisfy \edDP or its variant.
Machanavajjhala et al.~\cite{MKA+08} introduced a variant of \edDP called $(\epsilon,\delta)$-probabilistic differential privacy and showed that it is satisfied by a synthetic data generation method for the problem of releasing the commuting patterns of the population in the United States.  This notion is stronger than \edDP.  Korolova et al.~\cite{KKM+09} considered publishing search queries and clicks that achieves $(\epsilon,\delta)$-differential privacy. A similar approach for releasing query logs with differential privacy was proposed by G{\"o}tz et al.~\cite{GMW+09}.  These approaches apply the output perturbation technique in differential privacy to microdata publishing scenarios that can be reduced to histogram publishing at their core.
Blum et al.~\cite{BLR08} and Dwork et al.~\cite{DNR09} considered outputing synthetic data generation that is useful for a particular class of queries. These papers do not deal with the relationship between $k$-anonymization and differential privacy, or between sampling and $k$-anonymization.

Kifer and Lin~\cite{KL10} developed a general framework to characterize relaxation of differential privacy.  They identified two axioms for a privacy definition: Transformation Invariance and Privacy Axiom of Choice, which are satisfied by \bedDPS.  They did not consider the composability of these notions, which was our emphasis, as a clear understanding of the composability issues directs us what can and cannot be done with sampled dataset.

\section{Conclusions} \label{sec:conclusions}

We have answered the two questions we set out in the beginning of the paper.  We take the approach of starting from both $k$-anonymization and differential privacy and trying to meet in the middle.  On the one hand, we identify weaknesses in the $k$-anonymity notion and existing $k$-anonymization methods and propose the notion of safe $k$-anonymization to avoid these privacy vulnerabilities.  On the other hand, we try to relax differential privacy to take advantage of the adversary's uncertainty of the data.  The key insight underlying our results is that random sampling can be used to bridge this gap between $k$-anonymization and differential privacy.

We have explored both the power and potential pitfalls to take advantage of sampling in private data analysis or publishing.  Our results show that sampling, when used correctly, is a powerful tool that can greatly benefit differential privacy, as it creates uncertainty for the adversary.  Sampling can increase the privacy budget and error toleration bound.  Sampling also enables the usage of algorithms such as safe $k$-anonymization; however, this usage requires fresh sampling that is not used to answer any other query.  
An intriguing open question is whether there exist approaches other than sampling that can create uncertainty for the adversary, that can tolerate answering \eDP queries.

\bibliographystyle{abbrv}
\bibliography{Privacy,Ninghui}
\appendix

\section{Proofs}

This appendix includes proofs not included in the main body.

\subsection{Proof of Theorem~\ref{thm:amplification}} \label{app:amplification}

Theorem~\ref{thm:amplification}. \emph{Given an algorithm \AA that satisfies $\DPS{(\beta_1,\epsilon_1,\delta_1)}$, \AA also satisfies $\DPS{(\beta_2,\epsilon_2,\delta_2)}$ for any $\beta_2 < \beta_1$, where \\ $\epsilon_2=\ln\left(1+\left(\frac{\beta_2}{\beta_1}(e^\epsilon_1-1)\right)\right)$ and $\delta_2=\frac{\beta_2}{\beta_1}\delta_1$.}

\begin{proof}
We need to show that the algorithm $\AA^{\beta_2}$ satisfies $\DP{(\epsilon_2,\delta_2)}$.
Let $\beta = \frac{\beta_2}{\beta_1}$.  The algorithm $\AA^{\beta_2}$ can be viewed as first sampling with probability $\beta$, then followed by applying the algorithm $\AA^{\beta_1}$, which satisfies $\DP{(\epsilon_1,\delta_1)}$.

We use $\Lambda_{\beta}(D)$ to denote the process of sampling from $D$ with sampling rate $\beta$. Any pair $D,D'$ can be viewed as $D$ and $D_{-t}$, where $D_{-t}$ denotes the dataset resulted from removing one copy of $t$ from $D$.  For any $O$, let
$$Z=\Pr[\AA^{\beta_2}(D)\in O], \mbox{ and } X=\Pr[\AA^{\beta_2}(D_{-t}) \in O],$$
we want to show that
$$\left(Z \leq e^{\epsilon_2} X + \delta_2\right)  \;\wedge\; \left(X \leq e^{\epsilon_2}Z + \delta_2\right).$$
We have
$$Z=\sum_{T\subseteq D}\Pr[\Lambda_{\beta}(D)=T] \Pr[\AA^{\beta_1}(T)\in O],$$
$$X=\sum_{T\subseteq D_{-t}}\Pr[\Lambda_{\beta}(D_{-t})=T] \Pr[\AA^{\beta_1}(T)\in O].$$

To analyze $Z$, we note that all the $T$'s that resulted from sampling from $D$ with probability $\beta$ can be divided into those in which $t$ is not sampled, and those in which $t$ is sampled.  For a $T$ in the former case, we have
$$
\begin{array}{ll}
\Pr[\Lambda_{\beta}(D)=T] & = (1-\beta) \Pr[\Lambda_{\beta}(D)=T | t\mbox{ not sampled in }T]
\\
& = (1-\beta)\Pr[\Lambda_{\beta}(D_{-t})=T]
\end{array}
$$
For a $T$ in the latter case, we have
$$
\begin{array}{ll}
 \Pr[\Lambda_{\beta}(D)=T] & = \beta \Pr[\Lambda_{\beta}(D)=T | t\mbox{ sampled in }T]
 \\
 & = \beta \Pr[\Lambda_{\beta}(D_{-t})=T_{-t}].
\end{array}
$$
Hence we have
$$
{\scriptsize
\begin{array}{rl}
Z = & \sum_{T \subseteq D_{-t}}(1-\beta)  \Pr[\Lambda_{\beta}(D_{-t})=T] \Pr[\AA^{\beta_1}(T)\in O]
\\
 & + \sum_{T_{-t} \subseteq D_{-t}}\beta \Pr[\Lambda_{\beta}(D_{-t})=T_{-t}]  \Pr[\AA^{\beta_1}(T_{-t})\in O]
\end{array}
}
$$
Let
$$Y = \sum_{T'\subseteq D_{-t}}\Pr[\Lambda_{\beta}(D_{-t})=T'] \Pr[\AA^{\beta_1}(T'_{+t})\in O],$$
$$\mbox{then we have } Z = (1-\beta)X + \beta Y.$$


That $\AA$ satisfies $\DPS{(\beta_1,\epsilon_1,\delta_1)}$ means that  for each $T,O$
$$ \Pr[\AA^{\beta_1}(T_{+t})\in O] \leq e^{\epsilon_1} \Pr[\AA^{\beta_1}(T)\in O] + \delta_1$$
Hence we have
$$
{\scriptsize
\begin{array}{rl}
Y & \leq \sum_{T'\subseteq D_{-t}} \Pr[\Lambda_{\beta}(D_{-t})=T'] \left(e^{\epsilon_1} \Pr[\AA^{\beta_1}(T'\in O]+\delta_1\right),
\\
& = e^{\epsilon_1} \sum_{T'\subseteq D_{-t}} \Pr[\Lambda_{\beta}(D_{-t})=T'] \Pr[\AA^{\beta_1}(T')\in O]
\\
& \;\;\;+ \delta_1 \sum_{T'\subseteq D_{-t}} \Pr[\Lambda_{\beta}(D_{-t})=T']
\\
& = e^{\epsilon_1}  X + \delta_1.
\end{array}
}
$$
Hence we have
$$
\begin{array}{rrl}
Z &=&(1-\beta)X + \beta Y
\\
 & \leq & (1-\beta)X + \beta (e^{\epsilon_1}X +\delta_1)
 \\
 & \leq & (1-\beta + \beta e^{\epsilon_1}) X + \beta \delta_1.
 \\
 & = & e^{\epsilon_2} X + \delta_2.
\end{array}
$$

To show that $X \leq e^{\epsilon_2} Z + \delta_2$, we observe that $\AA$ satisfies $\DPS{(\beta_1,\epsilon_1,\delta_1)}$ means that
$$X \leq e^{\epsilon_1}  Y + \delta_1, \mbox{ and hence}$$
$$
Z= (1-\beta)X + \beta Y \geq (1-\beta) X + \beta e^{-\epsilon_1} (X-\delta_1),
$$
$$
\mbox{and }X \leq \frac{1}{1-\beta+\beta e^{-\epsilon_1}} Z + \frac{\beta e^{-\epsilon_1}}{1-\beta+ e \beta^{-\epsilon_1}} \delta_1
$$

We now show that
$$\frac{1}{1-\beta+\beta e^{-\epsilon_1}} \leq e^{\epsilon_2=\ln\left(1+\left(\frac{\beta_2}{\beta_1}(e^{\epsilon_1}-1)\right)\right)}=1+\beta(e^{\epsilon_1}-1) = e^{\epsilon_2}.
$$
$$
\begin{array}{ll}
 & \frac{1}{1-\beta(1-e^{-\epsilon_1})} \leq 1+\beta(e^{\epsilon_1}-1)
\\
\Leftrightarrow & 0 \leq (1+\beta(e^{\epsilon_1}-1))(1-\beta(1-e^{-\epsilon_1})) -1
\\
\Leftrightarrow & 0 \leq (e^{\epsilon_1}+e^{-\epsilon_1}-2)(\beta-\beta^2).
\end{array}
$$
Hence
$$\frac{\beta e^{-\epsilon_1}}{1-\beta+ e \beta^{-\epsilon_1}} \delta_1 \leq \beta e^{-\epsilon_1} e^{\epsilon_2} \delta_1 \leq \beta \delta_1 = \delta_2.$$
\end{proof}

\subsection{Proof of Theorem~\ref{thm:main}} \label{app:main}

Theorem~\ref{thm:main}: \emph{Any strongly-safe $k$-anonymization algorithm satisfies $\DPS{(\beta,\epsilon,\delta)}$ for any $0<\beta<1$, $\epsilon\geq - \ln(1-\beta)$, and $\delta = d(k,\beta,\epsilon)$, where the function $d$ is defined as
\begin{eqnarray*}
d(k,\beta,\epsilon) &= &
\max_{n: n \geq \left\lceil \frac{k}{\gamma} -1 \right\rceil }\sum_{j>\gamma n}^{n} f(j; n, \beta),
\end{eqnarray*}
where $\gamma = \frac{(e^\epsilon-1+\beta)}{e^\epsilon}$.
}

\begin{proof}
Let $\AA$ denote the algorithm, and $g$ be the data-independent generalization procedure in the algorithm.  For any dataset $D$, any tuple $t\in D$, and for any output $S$.  For any $\epsilon \geq - \ln(1-\beta)$, we show that the probability by which
\begin{equation}
 e^{-\epsilon} \leq \frac{\Pr[\AA(D)=S]}{\Pr[\AA(D_{-t})=S]} \leq e^{\epsilon}  \label{eqn:proof1}
\end{equation}
is violated is $\delta$.  Note that this is a stronger version of \edDP than the one in Definition~\ref{def:edDP}.  See~\cite{KS08} for relationship between the two.

Let $n$ be the number of $t'$ in $D$ such that $g(t')=g(t)$.  Let $j$ be the number of times that $g(t)$ appears in $S$.   Note that as the only difference between $D$ and $D_{-t}$ is that $D$ has one extra copy of $t$, we have.
$$
 \frac{\Pr[\AA(D)=S]}{\Pr[\AA(D_{-t}))=S]} = \frac{\Pr[\AA(D)\mbox{ has $j$ copies of g($t$)}]}{\Pr[\AA(D_{-t})\mbox{ has $j$ copies of g($t$)}]}
$$

Because any tuple that appears less than $k$ times is suppressed, either $j\geq k$, or $j=0$.
When $j=0$, we have
$$
\frac{\Pr[\AA(D)=S]}{\Pr[\AA(D_{-t})=S]} = \frac{F(k-1;n,\beta)}{F(k-1;n-1,\beta)} = \frac{\sum_{i=0}^{k-1} f(i; n, \beta)}{\sum_{i=0}^{k-1} f(i; n-1, \beta)}.
$$
Because $F(k-1;n,\beta)$ is always less than $F(k-1;n-1,\beta)$;\footnote{Let $X_i$'s be random variables that take the value $1$ with probability $\beta$, and $0$ with probability $1-\beta$.  $F(k-1;n-1,\beta)$ is the probability that the sum of $n-1$ such $X$'s $\leq k-1$, and $F(k-1;n,\beta)$ is the probability that the sum of $n$ such $X$'s is $\leq k-1$.} hence $\frac{\Pr[\AA(D)=S]}{\Pr[\AA(D_{-t})=S]}<e^{\epsilon}$.  Furthermore, we note that $\forall i \in [0..k-1]$, $\frac{f(i;n,\beta)}{f(i;n-1,\beta)}= \frac{n(1-\beta)}{n-i} \geq (1-\beta)$.  Hence $\frac{\Pr[\AA(D)=S]}{\Pr[\AA(D_{-t}))=S]} \geq (1-\beta)$.
Because $\epsilon \geq - \ln(1-\beta)$, we have $e^{-\epsilon} \leq 1-\beta $; hence under the case when $j=0$, inequality (\ref{eqn:proof1}) is satisfied.

When $j\geq k$, we have
$$
\frac{\Pr[\AA(D)=S]}{\Pr[\AA(D_{-t}\ref{thm:main}))=S]} = \frac{f(j;n,\beta)}{f(j;n-1,\beta)} =
\left\{\begin{array}{lr}
 \frac{n(1-\beta)}{n-j} & n \geq j
 \\
 1 & n<j.
 \end{array}
 \right.
$$

The choice of $n$ can be arbitrary because it is determined by the choice of $D$.  The value of $j$ is determined by the choice of $S$.  For some values of $j$, inequality~(\ref{eqn:proof1}) is violated.  We want to compute the probabilities of these bad $j$'s occurring.  From the above, we know 
when $j>n$, the outcome is good.  We now consider the bad outcomes when $j \leq n$.

Note that because $\epsilon\geq - \ln(1-\beta)$, we have $-\epsilon\leq \ln(1-\beta)$, and
$$\frac{n(1-\beta)}{n-j} > 1-\beta \geq e^{-\epsilon}.$$
Hence we only need to consider what $j$'s make $\frac{n(1-\beta)}{n-j} > e^\epsilon$.  This occurs when $j> \frac{(e^\epsilon-1+\beta)n}{e^\epsilon}$.  Let $\gamma = \frac{(e^\epsilon-1+\beta)}{e^\epsilon}$, then this occurs when $j > \gamma n$.

So far our analysis has shown that \ref{thm:main}a bad outcome $S$ for an input $D$ would satisfy the condition $j \geq k$ and $n \geq j > \gamma n$.  Now we need to compute the probability that $\AA(D)$ gives a bad outcome, and the probability that $\AA(D_{-t})$ gives a bad outcome.  The former is given below:
\begin{equation}
\max_n \sum_{j: \left(j\geq k \wedge j > \gamma n\right)}^n f(j; n, \beta) \label{eqn:errorprob}
\end{equation}
And the latter is
$$
\max_n \sum_{j: \left(j\geq k \wedge j > \gamma n\right)}^{n-1} f(j; n-1, \beta).
$$
As the latter is smaller than the former, we only need to bound the former.

Let $n_m = \left\lceil \frac{k}{\gamma} -1 \right\rceil$, we now show that when $n \leq n_m$, $\sum_{j: \left(j\geq k \wedge j > \gamma n\right)}^n f(j; n, \beta)$ increases when $n$ increases.  Note that the choice of $n_m$ satisfies the condition that $\gamma n_m < k$ and $\gamma (n_m+1) \geq k$.  Observe that when $n \leq n_m$, the condition $\left(j\geq k \wedge j > \gamma n \right)$ becomes $j\geq k$.  The function $\sum_{j: j\geq k}^n f(j; n, \beta)$ is monotonically increasing with respect to $n$.

When $n\geq n_m$, the condition $\left(j\geq k \wedge j > \gamma n\right)$ becomes $j>\gamma n$. (In fact, when $n=n_m+1$, the smallest $j$ to satisfy $j > \gamma n$ is $k+1$.)
Hence the error probability is bounded by
{\scriptsize
$$
\delta=d(k,\beta,\epsilon)=\max_{n: n \geq \left\lceil \frac{k}{\gamma} -1 \right\rceil }\sum_{j>\gamma n}^{n} f(j; n, \beta), \mbox{ where } \gamma = \frac{(e^\epsilon-1+\beta)}{e^\epsilon}.
$$
}

\end{proof}

\subsection{Proof of Theorem~\ref{thm:esafe}} \label{app:esafe}

Theorem~\ref{thm:esafe}:\emph{
Any $\epsilon_1$-safe $k$-anonymization algorithm satisfies $\DPS{(\beta, \epsilon, \delta)}$, where $\epsilon \geq -\ln(1-\beta)+\epsilon_1$, $\delta = d(k, \beta, \epsilon-\epsilon_1)=\max\limits_{n:n\geq\lceil \frac{k}{\gamma}-1\rceil}\sum_{j>\gamma n}^{n}f(j;n,\beta)$, $\gamma = \frac{(e^{\epsilon-\epsilon_1}-1+\beta)}{e^{\epsilon-\epsilon_1}}$.
}

\begin{proof}
Let $\mathcal{A}$ denote the $\epsilon_1$-safe $k$-anonymization algorithm. Here, we want to show that for any $\epsilon \geq -\ln(1-\beta)+\epsilon_1$, $D$, $t\in D$ and $S$,
\begin{equation}\label{eq}
e^{-\epsilon} \leq \frac{\Pr[\mathcal{A}(D)=S]}{\Pr[\mathcal{A}(D_{-t})=S]} \leq e^{\epsilon}
\end{equation}
is valid for probability at least $1-\delta$.
Let $\Lambda_{\beta}$ denote the process of binomial sampling the dataset $D$ with probability $\beta$. And let $G$ denote the set of all the possible outputs of $\mathcal{A}$'s subroutine $\mathcal{A}_m$. By definition, its subroutine $\mathcal{A}_m$ satisfies $\epsilon_1$-differential privacy,
\begin{equation*}
e^{-\epsilon_1} \leq \frac{\Pr[\mathcal{A}_m(D)=g]}{\Pr[\mathcal{A}_m(D_{-t})=g]} \leq e^{\epsilon_1}.
\end{equation*}
And, according to the proof of Theorem~\ref{thm:esafe}, for a fixed $g\in G$, the ratio $r(g) = \frac{\Pr[g(\Lambda_{\beta}(D))=S]}{\Pr[g(\Lambda_{\beta}(D_{-t}))=S]}$ equals
\begin{displaymath}
r(g) = \left\{\begin{array}{ll}
\frac{\sum_{i=0}^{k-1}f(i; n, \beta)}{\sum_{i=0}^{k-1}f(i; n-1, \beta)} \qquad &\mbox{if } j=0;\\ \frac{n(1-\beta)}{n-j} \qquad &\mbox{if } k\leq j\leq n
\end{array}\right.
\end{displaymath}
where $j$ is the number of copies of $g(t)$ in the output dataset $S$.
So, the differential privacy ratio (\ref{eq}) can be upper bounded,
$$
\begin{array}{rl}
 & \frac{\Pr[\mathcal{A}(D) = S]}{\Pr[\mathcal{A}(D_{-t}) = S]}
 \\
=& \frac{\sum_{g\in G}\Pr[\mathcal{A}_{m}(D) = g]\cdot \Pr[g(\Lambda_{\beta}(D))=S]}
{\sum_{g\in G}\Pr[\mathcal{A}_{m}(D_{-t}) = g]\cdot \Pr[g(\Lambda_{\beta}(D_{-t}))=S]}\\
\leq & \frac{e^{\epsilon_1}\sum_{g\in G}\Pr[\mathcal{A}_{m}(D_{-t}) = g]\cdot \Pr[g(\Lambda_{\beta}(D))=S]}
{\sum_{g\in G}\Pr[\mathcal{A}_{m}(D_{-t}) = g]\cdot \Pr[g(\Lambda_{\beta}(D_{-t}))=S]}\\
\leq & \frac{e^{\epsilon_1}r(g)\sum_{g\in G}\Pr[\mathcal{A}_{m}(D_{-t}) = g]\cdot\Pr[g(\Lambda_{\beta}(D_{-t}))=S]}
{\sum_{g\in G}\Pr[\mathcal{A}_{m}(D_{-t}) = g]\cdot\Pr[g(\Lambda_{\beta}(D_{-t}))=S]}\\
=& e^{\epsilon_1}r(g).
\end{array}
$$
The lower bound can be obtained in a similar way. So,
\begin{equation*}
e^{-\epsilon_1}r(g) \leq \frac{\Pr[\mathcal{A}^{\beta}(D) = S]}{\Pr[\mathcal{A}^{\beta}(D_{-t}) = S]}
\leq e^{\epsilon_1}r(g).
\end{equation*}
By the proof of Theorem ~\ref{thm:esafe}, the ratio $r(g)$ is bounded by $e^{ -(\epsilon - \epsilon_1)}\leq r(g) \leq e^{(\epsilon - \epsilon_1)}$. The probability that it is violated is the probability that inequality \eqref{eq} is violated.
In the $j=0$ case, $e^{ -(\epsilon - \epsilon_1)} \leq \frac{\sum_{i=0}^{k-1}f(i; n, \beta)}{\sum_{i=0}^{k-1}f(i; n-1, \beta)} \leq e^{(\epsilon - \epsilon_1)}$, since $\epsilon \geq -\ln(1-\beta)+\epsilon_1$.
And for the $k\leq j\leq n$ case, $\frac{n(1-\beta)}{n-j}> (1-\beta) \geq e^{-\epsilon+\epsilon_1}$. And only when $\frac{n(1-\beta)}{n-j}>e^{\epsilon-\epsilon_1}$ \big($j> \frac{n(e^{\epsilon-\epsilon_1}-1+\beta)}{e^{\epsilon-\epsilon_1}}$\big), inequality \eqref{eq} is violated. Let $\gamma = \frac{(e^{\epsilon-\epsilon_1}-1+\beta)}{e^{\epsilon-\epsilon_1}}$. The error probability $\delta$ is
\begin{displaymath}
\delta = d(k, \beta, \epsilon-\epsilon_1)=\max_{n:n\geq\lceil \frac{k}{\gamma}-1\rceil}\sum_{j>\gamma n}^{n}f(j;n,\beta),
\end{displaymath}
where $\gamma = \frac{(e^{\epsilon-\epsilon_1}-1+\beta)}{e^{\epsilon-\epsilon_1}}$.
\end{proof}

\end{document}